\pdfoutput=1
\documentclass{llncs}
\usepackage{amsfonts}
\usepackage{amsmath}
\usepackage{amssymb}
\usepackage{ifthen}
\usepackage{epsfig}
\usepackage{setspace}
\usepackage{xspace}
\usepackage{fullpage}

\newcommand{\infig}[3]{
\begin{figure}[t]
\centering
\includegraphics[width=#1in]{#2}
\caption{#3\label{fig:#2}}
\end{figure}}

\newtheorem{thm}[theorem]{Theorem}
\newtheorem{prp}[lemma]{Proposition}
\newtheorem{lem}[lemma]{Lemma}

\newtheorem{dfn}{Definition}

\newcommand{\cB}{{\cal B}\xspace}
\newcommand{\cH}{{\cal H}\xspace}
\newcommand{\cE}{{\cal E}\xspace}
\newcommand{\cR}{{\cal R}\xspace}
\newcommand{\cL}{{\cal L}\xspace}
\newcommand{\Vol}{\mathrm{Vol}}
\newcommand{\Org}{\mathrm{Org}}
\newcommand{\Dst}{\mathrm{Dst}}

\def\qed{\hspace*{\fill} $\Box$\par\medskip}

\def\<{\left<}
\def\>{\right>}

\def\Pr{\mathrm{Pr}}

\def\e{\mathrm{e}}

\newcommand\card[1]{\lvert #1 \rvert}
\def\rh{\widetilde{h}}
\def\RC{\widetilde{C}}
\def\RT{\widetilde{T}}

\title{Simple Random Walks on Radio Networks\\ {\small(Simple Random Walks on Hyper-Graphs)}}
\author{Chen Avin  \and Yuval Lando  \and  Zvi Lotker }
\institute{Department of Communication Systems Engineering \\
Ben-Gurion University of the Negev,  P.O.B 653, Beer-Sheva 84105,\\
Israel\\
\email{avin@cse.bgu.ac.il, lando@bgu.ac.il, zvilo@cse.bgu.ac.il}}

\date{}
\begin{document}
\maketitle

\begin{abstract}
In recent years, protocols that are based on the properties of random walks on graphs have found many applications 
in communication and information networks, such as wireless networks, peer-to-peer networks and the Web. 
For wireless networks (and other networks), graphs are actually not the correct model of the communication; instead hyper-graphs better capture the communication  over a wireless shared channel. 
Motivated by this example, we study in this paper random walks on hyper-graphs.  
First, we formalize the random walk process on hyper-graphs and generalize key notions from random walks on graphs.
We then give the novel definition of radio cover time, namely,  the expected time of a random walk to be heard (as opposed to visit) by all
nodes.
We then provide some basic bounds on the radio cover, in particular, we show that while on graphs the radio cover time is $O(mn)$,  in
hyper-graphs it is $O(mnr)$ where $n$, $m$ and $r$ are the number of nodes, the number of edges and the rank of the hyper-graph, respectively.
In addition, we define radio hitting times and give a polynomial algorithm to compute them.
We conclude the paper with results on specific hyper-graphs that model wireless networks in one and two dimensions. 
\end{abstract}
\textbf{Keyowrds:} Random walks, hyper-graphs, radio networks, cover time, hitting time, wireless networks.
\newpage
\pagenumbering{arabic}
\section{Introduction}
Random walks are a natural and thoroughly studied approach to
randomized graph exploration. A \emph{simple random walk} is a stochastic process that starts at one node of
a graph and at each step moves from the current node to an adjacent node chosen
randomly and uniformly from the neighbors of the current node. It can also been seen as uniformly selecting an adjacent edge
and stepping over it.
Since this process presents locality, simplicity, low-overhead (i.e, memory space) and
robustness to changes in the network (graph) structure, applications based on random-walk
techniques are becoming more and more popular in the networking
community. In recent years, different authors have proposed the use
of random walk for a large variety of tasks and networks; to name but a few: querying in wireless sensor and ad-hoc networks \cite{sadagoan03active,avin04efficient,alanyali06a-random-walk},
searching in peer-to-peer networks \cite{gkants04peer}, routing \cite{sergio02constrained,braginsky02rumor}, network connectivity \cite{broder89trading}, building spanning trees \cite{broder89generating}, gossiping \cite{kempe03gossip},  membership service \cite{bar-yossef06rawms}, network coding \cite{deb06algebraic} and quorum systems \cite{friedman08probabilistic}.

Two of the main properties of interest for random walks are \emph{hitting times} and the \emph{cover time}.
The hitting time between $u$ and $v$, $h(u,v)$, is the expected time (measured by the number of steps
or in our case by the number of messages) for a random walk starting at $u$ to visit node $v$ for the first time.
The {\em cover time} $C_G$ of a graph $G$ is the expected time taken by a simple random walk to visit all the nodes
in $G$.  This property is relevant to a wide range of algorithmic applications, such as searching, building a spanning tree and query processing
\cite{gkants04peer,wagner98robotic,jerrum97markov,avin04efficient}. Methods of bounding
the cover time of graphs have been thoroughly investigated \cite{matthews88covering,aldous89lower,chandra89electrical,broder89cover,zuckerman90lower}, with the major
 result being that cover time is always at most
polynomial for static graphs. More precisely, it has been shown by Aleliunas \emph{et al.} \cite{aleliunas79random} that   $C_G$ is always $O(mn)$,
 where $m$ is the number of edges in the graph and $n$ is the number of nodes.
Several bounds on the cover time of particular classes of graphs have been obtained, with many positive results:
for almost all graphs the cover time is of order $O(n\log n)$
\cite{chandra89electrical,broder89cover,jonasson98on,jonasson00planar,cooper03cover}.

While simple graphs are a good model for point-to-point communication networks, they do not capture well shared channel networks like wireless
networks and LANs. In wireless networks the channel is shared by many nodes; this, on the one hand, leads to contention, but on the other hand, can be
very useful for dissemination of information. When a node transmits on the shared channel, all other nodes that share the channel can receive or "hear"
the message. This situation, as noted in the past for other wireless network applications \cite{lun06network}, should be modeled as a (directed)
hyper-graph and not as a graph. Hyper-graphs are a generalization of graphs where edges are sets (or ordered sets) of nodes of arbitrary size. A graph is
a hyper-graph with the size of edges equal to two for all edges. For example, in wireless networks, there is a (directed) hyper-edge from each
transmitter to a set of receivers that can encode its message.

Motivated by the hyper-graph model for wireless networks,  in this paper we study random walks on hyper-graphs.
We extend the hitting time and cover time definitions to  the case of hyper-graphs, and in particular wireless networks.
We define the \emph{radio hitting time} from $u$ to $v$ as the expected number of steps for a random walk (to be defined formally later) starting at
$u$ to be heard by $v$ for the first time. Clearly, the radio hitting time is lower than the hitting time, so it will give a better bound on the
time to disseminate information between nodes using random walks on hyper-graphs.
But, while hitting times are well studied on graphs, it is not clear, at first sight, how to compute radio hitting times.
In a similar manner, we define the \emph{radio cover time} as the expected number of steps for a random walk to be heard by all the nodes in the
graph. Again, this will give a better bound on the time to spread information among all the nodes, for example, in a random-walk-based search or query.

To our surprise, we found that there is almost no previous work on random walks on hyper-graphs, and definitely not any theoretical work.
One exception is the experimental study on simple random walks on hyper-graph preformed by Zhou \emph{et al.} \cite{zhou07learning}.
In that work, the authors suggest using a simple random walk on hyper-graphs to analyze complex relationships among objects by learning and
clustering. Our work is general enough to make  contribute in this direction as well.

\subsection{Overview of Our Results and Paper Organization}
The paper contribution covers two main themes. In the first part (sections \ref{sec:models}--\ref{sec:radio}), we present formal
definitions for random walks on hyper-graphs. We extend known parameters and properties of random walks on graphs to the case of hyper-graphs; to
the best of our knowledge this the first rigorous treatment of this topic. We present a deep relation between a random walk on the set of vertices of
the hyper-graph to a random walk on the set of edges of the hyper-graph and between random walks on hyper-graphs and random walks on special bipartite
graphs. We study the undirected and directed cases and base all our definitions on the basic object that describes a hyper-graph, \emph{the incidence
matrix}. Moreover, we formally define the novel notion of radio hitting time and radio cover time, namely, the expected time for a specific node or all nodes to hear
a message carried by a random walk originating at a specific node.  This formalism will be essential tool in pursuing future research on random walks
on hyper-graphs.

The second theme of the paper is to provide algorithms and to prove bounds for the main properties of interest for random walks. In Section \ref{sec:computing},
Theorem \ref{thm:hit_e}, we provide an algorithm to compute the radio hitting time on hyper-graphs. Sections \ref{sec:speedup} and \ref{sec:general}
present general bounds on the radio cover time. In Section \ref{sec:speedup}, Theorem \ref{thm:rcspeedup}, we bound the cover time in terms of the
radio cover, and in Section \ref{sec:general} we generalize famous bounds on the cover times of graphs to radio cover times on hyper-graphs. In Theorem
\ref{thm:matthew} we extend Matthews' bound~\cite{matthews88covering} to radio hitting time and in Theorem  \ref{thm:mnr} we extend the fundamental
bound on the cover time of Aleliunas \emph{et al.} \cite{aleliunas79random} that bounds the cover time of graphs by $O(mn)$, to an $O(mnr)$ bound on
the radio cover time of hyper-graphs, where $n$ is the number nodes, $m$ is the number of edges and $r$ is the size of the maximum edge. Note that
while for graphs, $m$ is at most $n^{2}$, for hyper-graphs $m$ could be exponential, we show that even in this case the bound could be tight.
In Section \ref{sec:unit} we study hyper-graphs that model wireless radio networks. Theorems \ref{thm:1D} and \ref{thm:2D} bound the expected  time for
all nodes in the network to "hear" the message in 1-dimensional and 2-dimensional grids, respectively. These results capture some of the nice
properties of radio cover time; we show that as the size of edges increase the radio cover decreases. But while cover time cannot go below $n \log
n$, radio cover time can be much smaller, as a matter of fact when these grids become the complete graph the radio cover time is 1. 
Therefore these result will have a significant impact on the design of random-walk-based algorithms for wireless networks. Conclusions are then
presented in Section \ref{sec:conc}.
Due to the volume of the results and the space limitations, some of the proofs are presented in the appendix. The proofs, technical at times, also
contain interesting insights into the topic and are part of the full version of the paper.

\section{Models and Preliminaries}\label{sec:models}
We now present formal definitions of the involving objects, in some cases we follow definitions taken form PlanetMath.
\subsection{Definitions}
A (finite, undirected) \emph{graph} $G$ is an ordered pair of disjoint finite sets $(V, E)$ such that $E$ is a subset of the set $V \times V$ of unordered pairs of $V$.
The set $V$ is the set of \emph{vertices} (sometimes called nodes) and $E$ is the set of \emph{edges}. If $G$ is a graph, then $V = V(G)$ is the vertex set of $G$, and $E = E(G)$ is the edge set.
If $v$ is a vertex of $G$, we sometimes write $v \in G$ instead of $v \in V(G)$.

We follow with formal definitions for hyper-graphs.
\begin{dfn}[Hyper-graph]
A hyper-graph $\mathcal{H}$ is an ordered pair $(V, \mathcal{E})$, where $V$ is a set of vertices, and $\mathcal{E}  \subseteq 2^V$ is a set of hyper-edges between the vertices, i.e.,
each hyperedge $e \subseteq V$. 
The rank $r(\cH)$ of a hyper-graph $\cH$ is the maximum cardinality of any of the edges in the hyper-graph. If all edges have the same cardinality $k$, the hyper-graph is said to be $k$-uniform. A graph is simply a 2-uniform hyper-graph.
We use $\card{e}$ to denote the cardinality of the set $e$. For a hyper-edge $e \in \cE$, its degree is define to be $\delta(e) = \card{e}$.
The set $\mathcal{E}(v)=\{e\in \mathcal{E}: v \in e\}$ is the set of all edges that contain the vertex $v$. The \emph{degree} $d(v)$ of a vertex $v$ is the number of edges in $\cE(v)$, i.e., $d(v) = \card{\cE(v)}$.
$\mathcal{H}$ is $k$-\emph{regular} if every vertex has degree $k$.
The set of neighbors of a vertex $v$ is $N(v)=\{u \in V: \{v,u\} \subseteq e \in \cE\}$.
\end{dfn}

Let $V = \{v_1, v_2, ~\ldots, ~ v_n\}$ and $\mathcal{E} = \{e_1, e_2, ~ \ldots, ~ e_m\}$. Associated with any hyper-graph is the $n \times m$
\emph{incidence matrix} $W = (w_{ij})$ where
\[w_{ij} =
\begin{cases}
1 &\text{ if } ~ v_i \in e_j \\
0 &\text{ otherwise }
\end{cases}\]

Note that the sum of the entries in any column is the degree of the corresponding edge.
Similarly, the sum of the entries in a particular row is the degree of the corresponding vertex.
Let $D_{v}$ and $D_{e}$ denote the diagonal matrices of the degrees of the vertices and edges, respectively.



We change the standard definition of a directed hyper-graph 
and use the following definition (the difference between our definition and the standard is that we remove the condition that X,Y are disjoint.):
\begin{dfn}[Directed Hyper-graph]
A directed hyper-graph $\mathcal{H}$ is an order pair $(V, \mathcal{E})$, where $V$ is a set of vertices, and $\mathcal{E}$  is a set of hyper-arcs (i.e., directed hyper-edges). A
hyper-arc $e\in \mathcal{E}$ is an ordered pair $(U, W)$ where $U$ and $W$ are not empty subsets of $V$. The sets U and W are called the origin and
the destination of $e$ and are denoted as $\Org(e)$ and $\Dst(e)$, respectively.
\end{dfn}


We model wireless radio networks as a special type of directed hyper-graph. In radio networks, transmitters send messages on a broadcast-wireless
channel and can therefore be received by a set of receivers. We model this interaction as a directed hyper-edge with one origin (transmitter) and a set
of destinations (receivers) that do not include the origin.
\begin{dfn}[Radio Hyper-graph]
A \textbf{Radio Hyper-graph} is a directed hyper-graph $\cH = (V,\cE)$ in which for every hyper-arc $e \in \cE$, $\card{\Org(e)}=1$ and $\Org(e)
\notin \Dst(e)$.
\end{dfn}

A graph can be converted into a radio hyper-graph in the following natural way: $\cH(V,\cE)= \cR(G(V,E)) $ where for every $v \in G$ we create an edge
$e \in \cE$ for which $\Org(e) = v$ and $\Dst(e)=N(v)$. For example, unit disk graphs \cite{Clark1990unit} are a very popular graph model for wireless
networks; for a unit disk graph $G$,
 we may consider the radio hyper-graph $\cR(G)$, which we believe captures more accurately aspects of wireless networks.

A key notion in a  random walk is a path. If we wish to extend the definition of a simple random walk from graphs into hyper-graphs, we need to
understand the equivalent of path in hyper-graphs.

\begin{dfn}[Hyper-path]
A \emph{hyper-path} in a hyper-graph is a finite sequence of alternating vertices and hyper-edges, beginning and ending with a vertex,
$u_1,f_1,u_2,f_2,u_3,\dots ,f_{n-1},u_n$ where for $1 \le i < n$, $u_i \in V, f_i \in \cE$ and $u_n \in V$ such that every consecutive pair of vertices $u_i$ and $u_{i+1}$ are in $u_i, u_{i+1} \in e_i$.
A directed hyperpath is a hyperpath where $u_i = \Org(f_i)$ and $u_{i+1} \in Dest(f_i)$
\end{dfn}

\subsection{Random Walks on Graphs}
We recall the definition of a simple random walk on a graph $G=(V,E)$ and then modify it to a simple random walk on a hyper-graph.
A random walk on a graph is a Markov chain on the state space $V$ and \emph{probability transition matrix} $P$.
The location of the random walks is a function from discrete time to the set of nodes $V$; we denote this function by $X(t):\mathbb{N}
\rightarrow V$.
The walk starts at some fixed node $X(0)$, and at time step $t$ it moves on the edge connected to the node $X(t)$ to one of its neighbors $X(t+1)$.
Let $e(t)=\{X(t),X(t+1)\}$ be the edge the random walk traversed at time $t$. The random walk is called \emph{simple with self loops} when with
probability $\frac{1}{2}$ the walk stays in the same node and with probability $\frac{1}{2}$ the next node is chosen uniformly at random from the
set
of neighbors of the node $X(t)$, i.e., $P(v,u) = \frac{1}{2d(v)}$ if $\{u,v\} \in E$ and 0 otherwise. Note that the simple random walk chooses the
edge $e(t)$ uniformly from the set of edges connected to $X(t)$. The stationary distribution of a walk, if such exists, is the unique probability
vector $\pi$ s.t. $\pi=\pi P$. It is well known the for the simple random walk on a connected graph the stationary distribution $\pi$ is such that for
every $v
\in V$, $\pi(v) = \frac{d(v)}{2m}$ where $m=\card{E}$.

\section{Random Walks on Hyper-Graphs}\label{sec:rw-hyper}
\subsection{Random Walks on Undirected Hyper-Graphs}
We defined the \emph{simple random walk on hyper-graph} $\cH=(V,\cE)$ as a simple process of visiting the nodes of the hyper-graph in some random
sequential order. The walk starts at some fixed
node $X(0)$. Then, at each time step $t$ it moves on the hyper-edge $e(t) \in \cE$ connected to the node $X(t)$, i.e., $X(t)\in e(t)$.
The walk lands on one of the nodes in $e(t)$, formally $X(t+1)\in e(t)$. We chose $X(t+1)$ at random from $e(t)$.
The random walk is called \emph{simple} when the next edge $e(t)$ is chosen uniformly at random from $\cE(X(t))$,
and then $X(t+1)$ is chosen uniformly at random from $e(t)$.
The process of \emph{visiting} the nodes can, again, be described as a Markov chain with  the state space $V$ and transition matrix $P$.
The walk can move from vertex $v_i$ to the vertex $v_j$ if there is an hyper-edge that contains both vertices; therefore
the probability to move from vertex $v_i$ to $u_j$ is:
\begin{align}
P_{ij} = P(v_i,v_j) = \sum_{k=1}^{\card{\cE}}\frac{w_{ik} w_{jk}}{d(v_i)\delta(e_k)} =
\frac{1}{d(v_i)} \sum_{k=1}^{\card{\cE}}\frac{w_{ik} w_{jk}}{\delta(e_k)}
\end{align}
or alternatively the equation can be written
\begin{align}\label{eq:P}
P(v,u) = \sum_{\substack{e \in \cE,\\ \{v,u\} \subset  e}}\frac{1}{d(v)\delta(e)} =
\frac{1}{d(v)} \sum_{\substack{e \in \cE,\\ \{v,u\} \subset  e}}\frac{1}{\delta(e)}
\end{align}
Let $A$ be the \emph{vertex-edge transition matrix} $A=D_{v}^{-1}W$ and $B$ the
\emph{edge-vertex transition matrix} $B=D_{e}^{-1}W^T$. We can express $P$
in matrix form as $P=D_{v}^{-1}WD_{e}^{-1}W^T = AB$.

The stationary distribution of visiting the vertices is again well defined for this Markov chain and is known to be $\pi(v) = \frac{d(v)}{\Vol(V)}$ where $\Vol(V) = \sum_{u \in V} d(u)$.

Alternately, a random walk on a graph is a Markov chain on the edges of the graph. At each time step, the walk steps to a randomly chosen edge from the set of neighbors of the current edge. The state space of the chain is $\cE$ and the transition matrix is $Q$ and let $Y(t):\mathbb{N} \rightarrow \cE$ denote this process.
The probability to move from edge $e_i$ to $e_j$ is then:
\begin{align}
Q_{ij} = Q(e_i,e_j) = \sum_{k=1}^{\card{V}}\frac{w_{ki} w_{kj}}{\delta(e_i)d(v_k)} =
\frac{1}{\delta(e_i)} \sum_{k=1}^{\card{V}}\frac{w_{ki} w_{kj}}{d(v_k)}
\end{align}
or alternatively,
\begin{align}
Q(e,f) = \sum_{\substack{v \in V,\\ v \in e \cap f}}\frac{1}{\delta(e)d(v)} =
\frac{1}{\delta(e)} \sum_{\substack{v \in V,\\ v \in  e \cap f}}\frac{1}{d(v)}
\end{align}

In matrix form, we can express $Q$ as $Q=D_{e}^{-1}W^TD_{v}^{-1}W=BA$.

The stationary distribution of visiting edges is again well defined for this Markov chain and is known to be $\zeta(e) = \frac{\delta(e)}{\Vol(\cE)}$ where $\Vol(\cE) = \sum_{e \in \cE} \delta(e)$, note $\Vol(\cE) = \Vol(V)$.

Not that $X(t)$ and $Y(t)$ are coupled processes in the following sense: given $X(0)$ the distribution of $X(t), t >1$ can be expressed as follows:
\begin{align}
X(t)= X(0)P^t = X(0)AQ^{t-1}B
\end{align}
Similarly, given $Y(0)$, $Y(t)$ can be expressed in terms of $X(t)$:
\begin{align}
Y(t)= Y(0)Q^t = Y(0)BP^{t-1}A
\end{align}
Moreover $P$ and $Q$ share the same eigenvalues, since it is known that the non-zero eigenvalues of $AB$ and $BA$ are identical for matrices $A,B$
where $\dim(A) = \dim(B^T)$.

\subsection{Random Walks on Hyper-Graphs as Random Walks on Bi-Partite Graphs}
We can describe the simple random walks on an hyper-graph $\cH(V, \cE)$ as a simple random walk $Z(t)$, $t > 0$ on the following bi-partite graph
$\cB(\cH) = G(V \cup \cE, E_{\cB})$. The set of vertices contains both the vertices and the edges of the hyper-graph and the set of edges  $E_{\cB}
\subset V \times \cE$, in particular $(v_i, e_j) \in E_{\cB}$ iff $w_{ij} =1$ and edges are considered undirected. The adjacency matrix describing
$\cB(\cH)$ is the following $mn \times mn$ matrix:
\begin{align} A_{\cB} =
\begin{pmatrix} 0& W\\
W^T& 0
\end{pmatrix}
\end{align}
where $W$ is the incidence matrix of $\cH$.
The associated Markov chain is over the state space $V \cup \cE$ and the transition probability matrix $P_B$ is the following $mn \times mn$ matrix:
\begin{align}
P_{\cB} =
\begin{pmatrix}
0& A\\
B& 0
\end{pmatrix}
\end{align}
Clearly, if $X(0) = Z(0)$, then for $t > 0$, $X(t) = Z(2t)$, and the same holds for $Y(t)$ if $Y(0) = Z(0)$.

\subsection{Random Walks on Directed Hyper-Graphs}

We can define the walk in a similar manner for directed hyper-graphs. Let $\overrightarrow{W} = \Org(W)$ be the sub-matrix of $W$ where $\overrightarrow{w}_{ij} = 1$ only if $v_i \in \Org(e_j)$ and  $\overleftarrow{W}=\Dst(W)$ be the sub-matrix of $W$ where $\overleftarrow{w}_{ij} = 1$ only if $v_i \in \Dst(e_j)$. Let $\overrightarrow{d}(v)$ denote the sum of the row corresponding  to $v$ in $\Org(W)$ and $\overleftarrow{d}(e)$ the sum of the column
corresponding  to $e$ in $\Dst(W)$. Let $\overrightarrow{D}_{v}$ and $\overleftarrow{D}_{e}$ denote the diagonal matrixes for $\overrightarrow{d}(v)$ and $\overleftarrow{d}(e)$, respectively. Now, the transition matrix of entering an edge is  $\overrightarrow{A} = \overrightarrow{D}_{v} \overrightarrow{W}$ and the transition matrix of leaving an edge is  $\overleftarrow{B} = \overleftarrow{D}_{e} \overleftarrow{W}^T$. The transition matrix $P'$ for the Markov chain on the vertices is now $P'=\overrightarrow{A}\overleftarrow{B}$ and the transition matrix $Q'$ for the Markov chain on the edges is now $Q'=\overleftarrow{B}\overrightarrow{A}$. The transition matrix for the walk on
the bi-partitie graph is:
\begin{align}
P'_{\cB} =
\begin{pmatrix}
0& \overrightarrow{A}\\
\overleftarrow{B}& 0
\end{pmatrix}
\end{align}



\section{Radio Hitting Times and Radio Cover Time of Hyper-Graphs}\label{sec:radio}

\subsection{Hitting Times and Radio Hitting Times}
The hitting time $h(v,u)$ on a graph is the \emph{expected} time for a simple random walk starting at $v$ to reach $u$ for the first time.
When extending the notion of hitting time to hyper-graphs there are two basic approaches.
First we can talk about the expected time to \emph{visit} node $u$ for the first time starting at $v$; this naturally extend the results and techniques of hitting times on graphs to hyper-graphs . Second, motived by radio networks, we consider the \emph{radio hitting time}, the expected time it takes for $u$ to hear the message for the first time, i.e, the message was passed on an edge to which $u$ belongs.
We now define this formally: let $\mathcal{H}=(V, \mathcal{E})$ be a hyper-graph.  Let $U\subset V$ be a sub-set of nodes of $\mathcal{H}$. Let $v
\in V$ be the starting node of the random walk on the hyper-graph $\mathcal{H}$, and we define $\inf \emptyset:=\infty$.
\begin{dfn}[hitting time]
Let $T_{v,X}^{U}$ be the random variable that denotes the stoping time
$$T_{v,X}^{U}=\inf \{t\in \mathbb{N}: X(t) \in {U} ,X(0)=v \in V\}.$$
The \emph{(hyper) hitting time} from $v$ to ${U}$ is $h(v,{U})=E[T_{v,X}^{U}]$ (or for short $h_v^{U}$).
\end{dfn}
Note that the definition holds for both graphs and hyper-graphs. Let $h_{\max}$ be the \emph{maximum hitting time}: $h_{\max}=\max_{u,v \in V} h(v,u)$

\begin{dfn}[Radio hitting time]
Let $\RT_{v,X}^{U}$ be the random variable that denotes the stoping time:
$$\RT_{v,X}^{U}=\inf \{t\in \mathbb{N}: e(t) \in {U} ,X(0)=v \in V\}.$$
The \emph{radio hitting time} from $v$ to ${U}$ is $\rh(v,{U})=E[\RT_{v,X}^{U}]$ (or for short $\rh_v^{U}$).
\end{dfn}
Let $\rh_{\max}$ be the \emph{maximum radio hitting time}: $\rh_{\max}=\max_{u,v \in V} \rh(v,u)$
Note that for graphs the hitting times and radio hitting graphs are identical, but for hyper-graphs the radio hitting times are at most at hitting times.

\subsection{Cover Time and Radio Cover Time}
The cover time of a graph is the \emph{expected} time to visit all nodes in the graph, starting from the worst case. The definition for a graph is based on hitting time and extends naturally to hyper-graphs.
The cover time for a simple random walk on the (hyper-) graph starting from a node $v \in V$ is the random time $C_v$ that takes for the simple
random walk starting at $v$ to visit all nodes in $V$. The cover time of a graph is the maximum expected of all cover times. Formally,

\begin{dfn}[cover time]
Let $C_v$ be the random time for a random walk starting at $v$ to visit all the nodes
$C_v=\max_{u \in V} T_v^u.$
The cover time of the graph is:
$C=\max_{v \in V} E[C_v]$
\end{dfn}

The radio cover time, on the other hand, is defined to be as the time for all nodes to "hear" the message, i.e., for all nodes the message passed on
at least one edge they belong to. The definition is a natural extension of the cover time using the radio hitting time:
\begin{dfn}[Radio cover time]
Let $\RC_v$ be the random time for all the nodes to "hear" the random walk starting at $v$
$\RC_v=\max_{u \in V} \RT_v^u.$
The radio cover time of the graph is:
$\RC=\max_{v \in V} E[\RC_v]$
\end{dfn}

\subsection{Speedup of Radio Cover Time}
It is clear that the radio hitting time and radio cover time are faster than the hitting time and cover, respectively. We are interested in the speedup of radio, i.e., the ratio between hitting (cover) time and radio hitting (cover) time.
Let the  \emph{speedup} of radio hitting time and cover time for a hyper-graph $\cH$ be:
\begin{align*}
\text{hitting speedup: } S_{\cH}(\rh) = \frac{h_{\max}}{\rh_{\max}} & & \text{cover speedup: }S_{\cH}(\RC) = \frac{C}{\RC}
\end{align*}

\section{Computing the Radio Hitting Time}\label{sec:computing}
Computing the hitting time of Markov chains is well known. Recall our notation  $h_{v}^{{U}} = E[T_{v,X}^{{U}}]$ is the expected hitting time for a walk
starting at $v$ to hit a node in ${U}$ for the first time; then it can be compute as follows.
\begin{prp}[\cite{stirzaker05stochastic}]\label{hit time markov} The mean hitting times are the minimal non negative solution to:
\[
h_{v_{i}}^{{U}} =
\begin{cases}
0 &\qquad v_{i}\in {U}\\
1 + \displaystyle{\sum_{j} P_{ij} h_{v_j}^{U}} &\qquad v_{i} \notin {U}
\end{cases}
\] 
\end{prp}

Since the process $X$ has the coupled process $Y$, it is useful to define stopping time on the process $Y$.

\begin{dfn}[Y hitting time ]
For any $\mathcal{A}\subset \cE$. Let $T_{e,Y}^\mathcal{A}$ be the random variable that denotes the stoping time
$$T_{e,Y}^\mathcal{A}=\inf \{t\in \mathbb{N}: Y(t) \in \mathcal{A} ,Y(0)=e \in \cE\}.$$
The \emph{(Y hyper) hitting time} from $e$ to $\mathcal{A}$ is $h(e,\mathcal{A})=E[T_{e,Y}^\mathcal{A}]$ (or for short $h_e^\mathcal{A}$).
\end{dfn}
As for Markov chains sometimes we are not given a specific starting position but a distribution $\lambda_0$. In this case, we will defined the stoping
time to be the inner product of the indicator functions with the hitting time random vector. For example let $$\lambda_0=X(0)A.$$
Note that $\lambda_0$ is a probability distribution over $\cE$, and Y(0) is the corresponding random variable. In this case we define
the event $F_e=\{Y(0)=e\}$. Using these events, we can define the indicator functions $\mathbb{I}_{F_e}$. Hence

\begin{equation}\label{dis-hitting-time}
T_{\lambda_0,Y}^\mathcal{A}=\sum_{e\in \cE} \mathbb{I}_{F_e} T_{e,Y}^\mathcal{A}
\end{equation}
We note that we can do the same for the process $X$, but we will not do this in this paper.

When we try to computed the radio hitting time on hyper graphs it is not necessary for the random walk to visit/hit the set ${U}$ of
vertices. A vertex from the
set ${U}$ must receive (hear) the message from one of its neighbors. This process is best understood when we look at the process $Y(t)$. To compute the
radio hitting time of the set ${U}$, we will defined the set
$$\cE(U)= \{ e\in \cE : e\cap {U} \neq \emptyset \}.$$ We can use the connection between the $X$ process and the $Y$ process to compute
the radio hitting time. The next lemma shows that the radio hitting time between nodes can be formulated as the hitting time between sets of edges.
\begin{lem} For all $t\in \mathbb{N}$
$$\Pr\left[\RT_{v,X}^{U}\leq t\right]= \Pr\left[T_{\lambda(0),Y}^{\cE({U})}\leq t\right]
$$ where $\lambda_0=X(0)A.$
\end{lem}
\begin{proof}
We prove the lemma for the case $u={U}\in V$; the general case follows the same arguments. Note that in this case $\cE(v)$ is the set of all hyper-edges that contain the node $v$. From the definition of $\RT_{v,X}^u$, and $T_{Y(0),Y}^{\cE({U})}$ it follows that $\RT_{v,X}^{U},T_{\lambda(0),Y}^{\cE({U})}\in \mathbb{N}$. We therefore can prove the lemma by induction.

For the base of the induction assume that $t=0$; in this case it follows that $u=v$ and the lemma follows.

Assume the lemma is true for all $t<k$, we prove the lemma for the case $t=k$. To prove the induction step, we have to prove that

$$\Pr\left[\RT_{v,X}^{U}= k\right]= \Pr\left[T_{\lambda(0),Y}^{\cE({U})}= k\right]
$$
Since the process $X,Y$ are coupled processes, it follows from the induction hypothesis that both processes did not radio hit the node $u$ before time $k$.

Now assume that $\RT_{v,X}^{U}= k$; this means that $u\in e(k)$. Therefore it follows from the definition of $\cE(u)$ that the hyper-edge $e(k)$ is a an element  $ e(k) \in \cE(u) $. Therefore by the definition of the stopping time $T_{\lambda(0),Y}^{\cE({U})}$ it follows that $T_{\lambda(0),Y}^{\cE({U})}=k$.

For the other direction, assume that $T_{\lambda_0,Y}^{\cE({U})}=k$; in this case using the induction hypothesis and the fact that both process are coupled processes, it follows that $\RT_{v,X}^{U}> k-1$. Moreover, since $T_{\lambda_0,Y}^{\cE({U})}=k$, it follows that $e(k)\in \cE(u)$. Therefore $\RT_{v,X}^{U}= k$, and the lemma follows.
\qed
\end{proof}
Now we can use equation \ref{dis-hitting-time} together with the previous lemma and compute the radio hitting time for a node $u$ starting
from a node $v$, $\rh_{v}^{u}$. We need first to solve the following linear system (see the linear system in Theorem \ref{thm:hit_e}) and then take a convex
combination of the the variables $\{h_e^{\cE({U})}: e\in \cE\}$ according to $\lambda_0$.
\begin{thm}\label{thm:hit_e} The radio hitting time $\rh_{v}^{u}$ is:
\begin{align}
\rh_{v}^{u} = E[\RT_{v,X}^u] = E[T_{\lambda_0,Y}^{\cE(u)}] = E_{\lambda_0}[T_{e,Y}^{\cE(u)}]
\end{align}
where $\lambda_0 = X(0)A$ and the Y hitting times are the minimal non negative solution to:
\[
h_{e_i}^{\cE(u)} =
\begin{cases}
0 &\qquad e_{i}\in \cE(u)\\
1 + \displaystyle{\sum_{j} Q_{ij} h_{e_{j}}^{\cE(u)}} &\qquad e \notin  \cE(u)
\end{cases}
\] 
\end{thm}


%



\section{The Speed-up of Radio Cover Time}\label{sec:speedup}
Clearly, the radio cover time is at most the cover time of the graph, but how much smaller it could be? We now show that it cannot be too small and the speedup of the radio cover time is bounded by $O(n \log n)$. This results in tight since there are graph for which the speedup is $\Omega(n\log n)$.

\begin{thm}\label{thm:rcspeedup}
Let $\cH=(V,\cE)$ be a hyper-graph.
Then $C\le O((n\log n) \RC)$ so $S_{\cH}(\RC)=O(n \log(n))$.
\end{thm}

\begin{proof}
Assume we start at node $s(1)$, i.e. $X(0)=s(1)$. Denote the first radio cover time by $\RC_{s(1)}^{(1)}$. We defined using induction
$s(i+1):=X\left(\RC_{s(i)}^{(i)}+1\right)$. Denote the $\RC_{s(i)}^{i+1}$-radio cover time to be the first time the process $X$ do a complete radio
cover time after the time $\RC_{s(i)}^{(i)}.$ Clearly, by the definition of the radio cover time and the linearity of the expectation for all $i\in
\mathbb{N}$, it follows that $E[\RC_{s(i)}^{i+1}]\leq (i+1) \RC.$
For each time we complete a radio cover all vertices in $V$, will have heard the random walks at least once. Assume the vertex $v$ heard
the random walk at time
$t_v^i$ for the first time in the $i$ radio cover. At each time, the probability that the random walk will land on the vertex $v$ at time $t_v^i$ is at
least $\frac{1}{n}$ i.e., $\Pr[X(t_v^i)=v]\geq \frac{1}{n}$. Now the proof follows from the coupon collector argument.
\qed
\end{proof}
We note that the previous theorem is tight. Consider a hyper-graph with $V=\{1,2,...,n\}$ nodes and one single hyper-edge $\cE=\{V\}$. 
It is clear
that the radio cover time of this hyper-graph is one. While by the coupon collector argument the cover time for this hyper-graph is $O(n \log n).$
\section{General Bounds for the Radio Cover Time}\label{sec:general}
%
Our first general bound on the cover time in an extension to Matthews' bound~\cite{matthews88covering} on the cover time.
\begin{thm}\label{thm:matthew}
Let $\cH=(V,\cE)$ be a hyper-graph, then:
\begin{align*}
C=O(h_{\max} \log n)  \text{ and } \RC=O(\rh_{\max} \log n)
\end{align*}
where $\card{V}=n$
\end{thm}
The first part of the theorem follows directly from Matthews' bound for the cover time of time homogenous strong Markov process~\cite{matthews88covering}:
For any reversible Markov chain on a graph ~$G$,
\begin{displaymath} h_{min}\cdot H_n \:\: \le \:\: C \:\: \le \:\: h_{\max} \cdot H_n
\end{displaymath}
\noindent where $H_k = \ln(k) + \Theta(1)$ is the k-th harmonic number.
We now prove the second statement of theorem \ref{thm:matthew}.
\begin{lem}
Let $\cH=(V,\cE)$ be a hyper-graph, then: $$\RC=O(\rh_{\max} \log n)$$ where $\card{V}=n$
\end{lem}
\begin{proof}
The proof technique is essentially identical to a generalization of Matthews' theorem to parallel random walks presented
in~\cite{alon08many}.
Recall, for any two vertices~$u,v$ in~$\cH$, $\rh(u',v')\le \rh_{\max}$.
By Markov inequality, $\Pr[$a random walk of length $\e \cdot \rh_{\max}$ starting from $u$ does not radio hit $v ] \le 1/\e$.
Hence for any integer $i>1$, the probability that a random walk of length $\e \cdot i \cdot \rh_{\max}$
does not radio hit~$v$ is at most~$1/\e^i$. (We can view the walk as~$i$ independent trials
to radio hit~$v$.)  Set $i=\lceil (2\ln n) \rceil$. Then the
probability that a random-walk of length $\e \cdot i \cdot \rh_{\max}$ does not radio hit~$v$
is at most $1/(n^{2})$. Thus with probability at least $1-(1/n)$ a random-walk
radio hits all vertices of~$\cH$ starting from~$u$. Together with theorem \ref{thm:rcspeedup},
we can bound the radio cover time of~$\cH$ by
\begin{align*}
\RC &\le \e \cdot i \cdot \rh_{\max} (1-\frac{1}{n}) + (n\log n) \cdot \rh_{\max}\frac{1}{n} \\
&\le O(\rh_{\max} \log n)
\end{align*}
The theorem follows. \qed
\end{proof}
Our second general bound for the cover time of hyper-graphs is a generalization of the fundamental result of Aleliunas \emph{et al.}
\cite{aleliunas79random} ,which bounds the cover time of a simple graph by $O(mn)$.
\begin{thm}\label{thm:mnr}
Let $\cH=(V,\cE)$ be a hyper-graph, then $$\RC \le C =O(m\cdot n \cdot r(\cH))$$ where $\card{V}=n$, $\card{\cE}=m$ and $r(\cH)$ is the rank of $\cH$.
\end{thm}

\begin{proof}
We use the bi-partite graph $B(\cH)=(V\cup \cE,E_B)$. We call the nodes of $B(\cH)$ that correspond to the nodes of the hyper-graph the left part of $B(\cH)$, and a nodes of $B(\cH)$ that corresponding to the hyper-edges of the hyper-graph the right part of $B(\cH)$. Denote the number of nodes in $B(\cH)$ by $n'$ and the number of edges in $B(\cH)$ by $m'$.
Since a node in $B(\cH)$ is node in the the hyper-graph $\cH$ or an hyper-edge in $\cH$, it follows that the total number of nodes in $B(\cH)$ is $n'=n+m$. We can bound the number of edges in $B(\cH)$ by $m' \leq m \cdot r(\cH)$; this follows since each hyper-edge is replaced by no more than $r(\cH)$ edges in $B(\cH)$. Clearly, the graph $B(\cH)$ is connected if and only if the hyper-graph $\cH$ is connected.

Let $T$ be a minimum (with respect to the number of nodes in $T$) tree that contains the entire left part of $B(\cH)$. Observe that tree $T$ exists since $B(\cH)$ is connected and finite. Since all nodes in the right part of $B(\cH)$ have a degree larger or equal to two it follows that the number of edges in $T$ is less than $2n$.

We use equation~\ref{eq:Ch1} to calculate the sum of commute time on the tree $T$.
Since the number of edges in $T$ is less than $2n$ it follows that $T$ will have no more than $n-1$ nodes that are belonging to the right part. Those nodes are crossposting to an hyper-edges on the hypergraph $\cH$.
therefore the sum of commute times on the tree is $2m' \cdot 2n \leq 4mn\cdot r(\cH)$. Note that every step on the hyper-graph is two steps on the bi-partite graph, and therefore the cover time of
the hyper-graph is no more than $2mn \cdot r(\cH)$.
\qed
\end{proof}
The above result is tight in the sense that we can show the following (proof in the appendix):
\begin{lemma}\label{lem:lowermnr}
For $c\in \mathbb{N}, c \le n/2$ there exist a hyper-graph $\cH=(V,\mathcal{E}),|V|=n,|\mathcal{E}|=m$, $r(\cH)=c$ with
$C \ge \RC \ge \rh = \Omega(mn\cdot c)$.
\end{lemma}

\section{The Radio Cover Time of Radio Hyper-Graphs in 1 and 2 Dimensions}\label{sec:unit}
Next we extend the notion of a line into a $k$-uniform hyperline, $\cH_1^k$, i.e., all the hyper edges have the same cardinality $c$,  the vertex set is $\mathcal{V}=\{1,2,...,n\}$, and the hyper-edges are $\cE=\{[1,k],[2,k+1],...[n-k,n]\}$, where $[a,b]=\{a,a+1,...,b\}$. In this case, the radio cover time is equal to the maximum radio hitting time. The next theorem computed and upper bound on the hitting time.

\begin{thm}\label{thm:1D}
For $2 \le k \le n$, Let $\cH_1^k$ be a $k$-uniform, $1$-dimensional mesh radio hyper-graph, then
$$\RC \le \frac{n^2}{\frac{1}{3} k^2+ \frac{1}{2} k+ \frac{1}{6}} =O\left ( \frac{n^2}{k^2}\right )$$
\end{thm}

\begin{proof}
The proof is based on martingale. We remind the reader that a Markov chain $(M_n)_{n\geq 0}$ is a martingale if:
$$E[M_{n+1}|M_{n}]=M_{n}.$$
Consider our Markov random walk $(X_t)_{t\geq 0}$ on the nodes of the infinite $k$-hyperline. 
We claim that $X_t$ is martingale. Observe
that $$M_t=X^{2}_{t}-\left(\frac{1}{3} k^2+ \frac{1}{2} k+ \frac{1}{6} \right) t$$ is also martingale.
Let $T=\inf\{t\geq 0:X_t\geq a \bigvee X_t \leq -a \in \mathbb{N}\} $. By the definition of $T$ it is clear that $T$ is a random stopping time and
that for all $a<\infty$, $T<\infty $. Therefore it follows form the standard theorem on martingale that $E[M_T]=E[M_0]=0$. Moreover,
$$0=E[M_0]=E[M_T]=E[X_t^2]-E\left[\left(\frac{1}{3} k^2+ \frac{1}{2} k+ \frac{1}{6} \right) T\right].$$
Using a simple algebra manipulation it follows that
$$\frac{a^2}{\frac{1}{3} k^2+ \frac{1}{2} k+ \frac{1}{6}} \leq E[T] \leq \frac{(a+k)^2}{\frac{1}{3} k^2+ \frac{1}{2} k+ \frac{1}{6}}.$$
\qed
\end{proof}
For the 2-dimensional case we prove that the radio cover time decreases with the size of the edges:
\begin{thm}\label{thm:2D}
For $4 \le k \le n$, Let $\cH_2^k$ be a $k$-uniform, $2$-dimensional mesh radio hyper-graph, then:
$$\RC=O\left( \frac{n}{k}\log\frac{n}{k}\log n \right)$$
\end{thm}

This result is significant for the design of many algorithms in wireless networks. For example, it implies that the time to reach all nodes is less
than
linear when the size of the edges is $O(\log^2 n)$. We believe that the bound is not tight and further improvement is feasible. The proof is using
the strong symmetry of the graph and using the following lemmas that are valuable in their own right.

\subsection{The Radio Cover Time of 2-dimensional Mesh Radio Hyper-graphs}
\begin{dfn}
A \textbf{$k$-hop, 2-Dimensional Mesh Radio Hyper-Graphs} denoted as $\cH_{2,k}$ is a $\sqrt{n} \times \sqrt{n}$
2-dimensional grid of nodes located at the torus, where each node $v$ is connected via a directed hyper-edge to all the nodes that are at most at
$L_{1}$-distance $k$ away from it.
\end{dfn}

We will bound the radio cover time of $\cH_{2,k}$ and this immediately implies the result of Theorem \ref{thm:2D} since $\cH_{2,k}$ is a 2k(k+1)-uniform $2$-dimensional mesh radio hyper-graph. To bound the radio cover time we will first bound
the maximum radio hitting time of $\cH_{2,k}$, and then use Theorem \ref{thm:matthew} to obtain the desired result.
For a hyper-graph $\cH(V,\cE)$, let $G_{\cH}(V,E)$ be the graph for the simple random walk $X(t)$ on $\cH$ with the transition probability $P$.
We will use $G_{\cH_{2,k}}$, or $G$ for short, to prove our results.
Note that $G$ is undirected since if $u$ belongs to the edge of $v$ in $\cH_{2,k}$ then $v$ belongs to the edge of $u$ as well; moreover, since every
node has only one hyper-edge, the random walk on $G$ is a simple random walk. This implies that the electrical network of $G$ consists of 1 ohm resistors.
We will use the strong symmetry of $G$, in particular the facts that $G$ is $d$-regular and vertex-transitive\footnote{a graph $G$ is vertex-transitive
if its automorphism group acts transitively upon its vertices, i.e, for every two vertices $u,v$ there is a automorphism  $f$ s.t. $f(u)=v$}.
\begin{lem}
For a transitive (hyper) graph, the hitting time can be expressed as:
$$
h(u,v) = h(u, N(v)) + h(w,v)
$$
where $N(v)$ is the set of neighbors of $v$ and $w$ is any neighbor of $v$.
\end{lem}

\begin{proof}
We can write $T_{u}^{v}$ as follows:
$$
T_{u}^{v} = T_{u}^{w} + T_{w}^{v}
$$
where $w \in N(v)$ is the first neighbor of v that the walk reaches; clearly, the walk must always reach a neighbor of $v$ before $v$. Then
$$
h(u,v) = E[T_{u}^{v}] = E[T_{u}^{w} + T_{w}^{v}] = E[T_{u}^{w}] + E[T_{w}^{v}] = h(u, N(v)) + \sum_{x\in N(v)} E[T_{w}^{v} \mid w=x]\Pr(x)
$$
where $\Pr(x)$ is the probability that $x$ is the first neighbor of $v$ reached by a walk starting at $u$.
Now, for transitive graphs $h(w,v) = h(w',v)$ if $w,w' \in N(v)$ by the transitivity of $w$ and $w'$ and the result follows.
\qed
\end{proof}

On radio hyper-graphs we can use the above results to bound the the radio hitting time.
\begin{lem}
For a transitive radio hyper-graph $\cH$ the radio hitting time is:
\begin{align}
\rh(u,v) &= h(u,v) - h(w,v) \label{eq:rh1} \\
           &= m(R_{uv} - R_{wv}) \label{eq:rh2}
\end{align}
where $w$ is any neighbor of $v$ and $R_{xy}$ is the effective resistance between $x$ and $y$.
\end{lem}
\begin{proof}
Eq \ref{eq:rh1} follows from the fact that, on radio hyper-graph, if a message reaches a neighbor $w$ of $v$, in the next step $v$ will hear the
message since $w$ can transmit on a single edge. Eq. \ref{eq:rh2} follows from the fact that on transitive graphs $h(x,y)=h(x,y)$ for any pair of
nodes and therefore using Eq. \ref{eq:Ch1} we have $h(x,y) = C_{xy}/2 = mR_{xy}$.
\qed
\end{proof}

To bound $\rh(u,v)$ we now prove upper bound on $R_{uv} - R_{wv}$. We do so by giving an upper bound for $R_{uv}$ and a lower bound for $R_{wv}$ where
$w$ is a neighbor of $v$. We present an upper bound for $R_{uv}$ using the method of \emph{unit flow}, in particular, we construct a legal unit flow
from $u$ to $v$ and by the \emph{Thomson Principle} \cite{doylesnell84} the power of the flow is an upper bound for $R_{uv}$. The main property of the flow (as opposed to a
very similar flow construction given in \cite{avin07on-the-cover}) is that both nodes $u$ and $v$ use all their edges in the flow with a flow $\frac{1}{d}$ on each
such edge. Let $d(u,v)$ denote the minimum distance in hops from $u$ to $v$ in $G$, then for every two nodes we have the following bound:

\begin{lem}\label{lem:ruv}
For any two nodes $u$ and $v$, the effective resistance $R_{uv}$ in $G_{\cH_{2,k}}$ is:
\begin{align}
R_{uv} \le \frac{2}{d} + O\left(\frac{\log(d(u,v))}{d^{2}}\right) \label{eq:ruv}
\end{align}
where $d= \Theta(k)$ s the degree of the uniform hyper-graph.
\end{lem}
\begin{proof}
We build a unit flow from $u$ to $v$ in the following way. Consider the shortest line between $u$ and $v$, we consider in the flow  only nodes
that lay inside a square that this line is on its diagonal, as shows in Fig. \ref{fig:figs/flow_A}. We divide this big square into small bins such
that $u$ (and $v$) are in the center of a minimal square that covers all the neighbors of $u$ and is build from 16 bins, see Fig.
\ref{fig:figs/flow_A} for clarification.
Note that in each bin we have $\Theta(d/16)$ nodes and let $d'$ denote the number of nodes in the bin with the minimum number
of nodes. Every node in a bin has an edge to all the nodes it the 9 bins around
its bin. The number of edges between two adjacent bins is $\theta(d^{2})$. The basic idea of the follow, similar to flows in \cite{avin07on-the-cover}, is to build the
flow in layers. In our flow, layers will increased linearly  (see Fig. \ref{fig:figs/flow_B}) until the middle point on the shortest line from $u$ to
$v$ and then will decrease linearly until reaching $v$. Layer $l$ will consist $l$ bins $(B_{l}^{1}, B_{l}^{2}, \dots B_{l}^{l})$, and in each bin $d'$
nodes will participate in the flow to a total of $ld'$ nodes in the layer.
Every layer will forward a unit flow to the next layer, in every layer the unit flow is divided uniformly between all the nodes of the layer.
The number of edges between layer $l$ and $l+1$ will be $ld'^2$ and therefore each edge between the layers carries $1/ld'^2$ of the unit flow.
The uniformity  of the flow between layers is guaranteed by each bin in layer $l$ selecting $\lceil (l+1)d'/l \rceil$ nodes  in layer $l+1$. Since
the two sets of nodes form a complete bi-partite graph the uniformity is guaranteed. An example of this kind of matching is given in Fig.
\ref{fig:figs/flow_C}.
Our first layer will consist of 8 bins  $(B_{8}^{1}, \dots B_{8}^{8})$ as presented in Fig. \ref{fig:figs/flow_B}. Therefore if there are $L$ layers
until the midpoint the power of the flow in this edges is
\begin{align}
\sum_{l=8}^{L} \sum_{e \in l,l+1}\frac{1}{(ld'^2)^{2}} &= \sum_{l=8}^{L} \frac{ld'^2}{l^{2}d'^{4}} \\
&= \frac{1}{d'^{2}}\sum_{l=8}^{L} \frac{1}{l} \\
&= O(\frac{\log L}{d^{2}})
\end{align}
From the midpoint the flow will be in decreasing layers, mirroring the first part of the flow until the last layer consisting as well 8 bins.
The first and last part of the flow is getting the unit flow from $u$ to the first layer and from the last layer to $v$.
This is explained with the help of Fig. \ref{fig:figs/flow_B}. First, node $u$ uses all of its $d$ edges and each carries $1/d$ of the flow.
Second, every neighbor of $u$ has a flow strategy that depends on its location in the 16 bins around $u$. This strategy guarantees the every bin in the
first layer will get $1/8$ of the flow and the unit flow will be divided uniformly among all the nodes in the first layer.
As before between bins each edge carries $\Theta(\frac{1}{d^{2}})$ of the flow. The number of edges (note including the edges adjacent to $u$) that
carries the flow to the first layer is less than $(1+2+3+4+5+6+7)d^{2}$, therefore the power in this part of the flow is
$O(24d^{2}/d^{4})=O(1/d^{2})$. A symmetric flow is when the flow goes from the last layer to $v$.
The total power of the flow is then the power over the following edges:
\begin{enumerate}
\item from $u$ to $N(u)$.
\item from $N(u)$ to the first layer.
\item from the first to the last layer.
\item from the last layer to $N(v)$.
\item from $N(v)$ to v.
\end{enumerate}
Putting the numbers together the result follows:
\begin{align*}
P(c) &= \frac{d}{d^{2}} + O(\frac{1}{d^{2}}) + O(\frac{\log L}{d^{2}}) + O(\frac{1}{d^{2}}) + \frac{d}{d^{2}} \\
&=\frac{2}{d} + O(\frac{\log d(u,v)}{d^{2}})
\end{align*}
\qed
\end{proof}

We now give a simple lower bound for $R_{wv}$ when $(w,v)$ is an edge in $G$
\begin{lem}
For any two nodes $w$ and $v$ such that the edge $(w,v) \in G_{\cH_{2,k}}$
\begin{align}
R_{wv} \ge \frac{2}{d +1} \label{eq:rwv}
\end{align}
\end{lem}

\begin{proof}
We lower $R_{wv}$ by the short/cut principal \cite{synge51,doylesnell84}, namely, shorting any two nodes in $G$ only decreases the resistance. We short all the nodes of the
graph (but $w$ and $v$) into one node called $x$. Since  $d(w)=d(v)=d$ and there is an edge $(w,v)$ the resulting graph has in addition $d-1$
parallel edges from $w$ to $x$ and from $x$ to $v$. In this graph the resistance between $w$ and $v$ is $\frac{2}{d +1}$ and the results follows.
\qed
\end{proof}

We now have everything to bound the maximum radio hitting time of $\cH_{2,k}$
\begin{lem}\label{lem:hitting2D}
The maximum radio hitting time, $\rh_{\max}$, of $\cH_{2,k}$ is
$$
\rh_{\max} = O(\frac{n}{d}\log \frac{n}{d})
$$
where $d=\Theta(k)$ is the degree of the uniform hyper-graph.
\end{lem}
\begin{proof}
From Equations \ref{eq:rh2}, \ref{eq:ruv} and \ref{eq:rwv} :
\begin{align*}
\rh(u,v) &\le nd\left(\frac{2}{d} + O\left(\frac{\log(d(u,v))}{d^{2}}\right)  -\frac{2}{d +1} \right) \\
           &= nd \left(\frac{2}{d(d+1)} + O\left(\frac{\log(d(u,v))}{d^{2}}\right) \right)\\
           &=\frac{n}{d+1} + O\left(n\frac{\log(d(u,v))}{d}\right) \\
           &=O\left(\frac{n}{d}\log(d(u,v))\right)
\end{align*}
The results now follows since the most distant pair of nodes are at distance $d(u,v) = \frac{\sqrt{n}}{d}$.
\qed
\end{proof}

\infig{2}{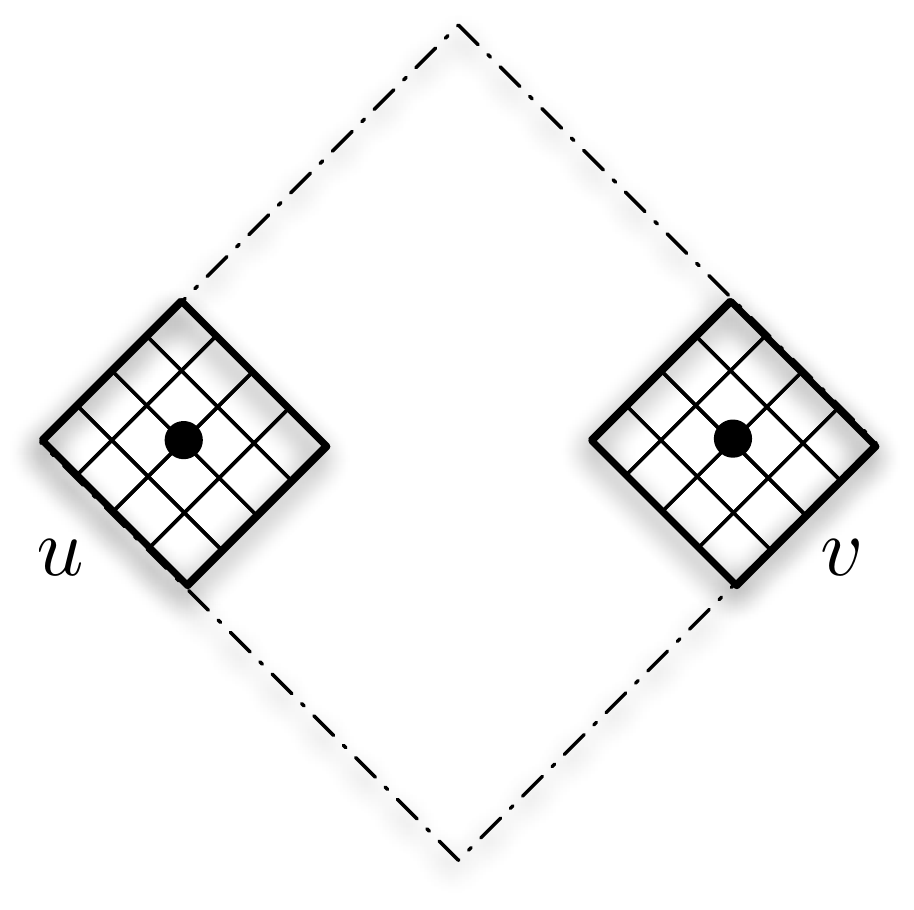}{The nodes considered in a flow between $u$ and $v$.}
\infig{1.5}{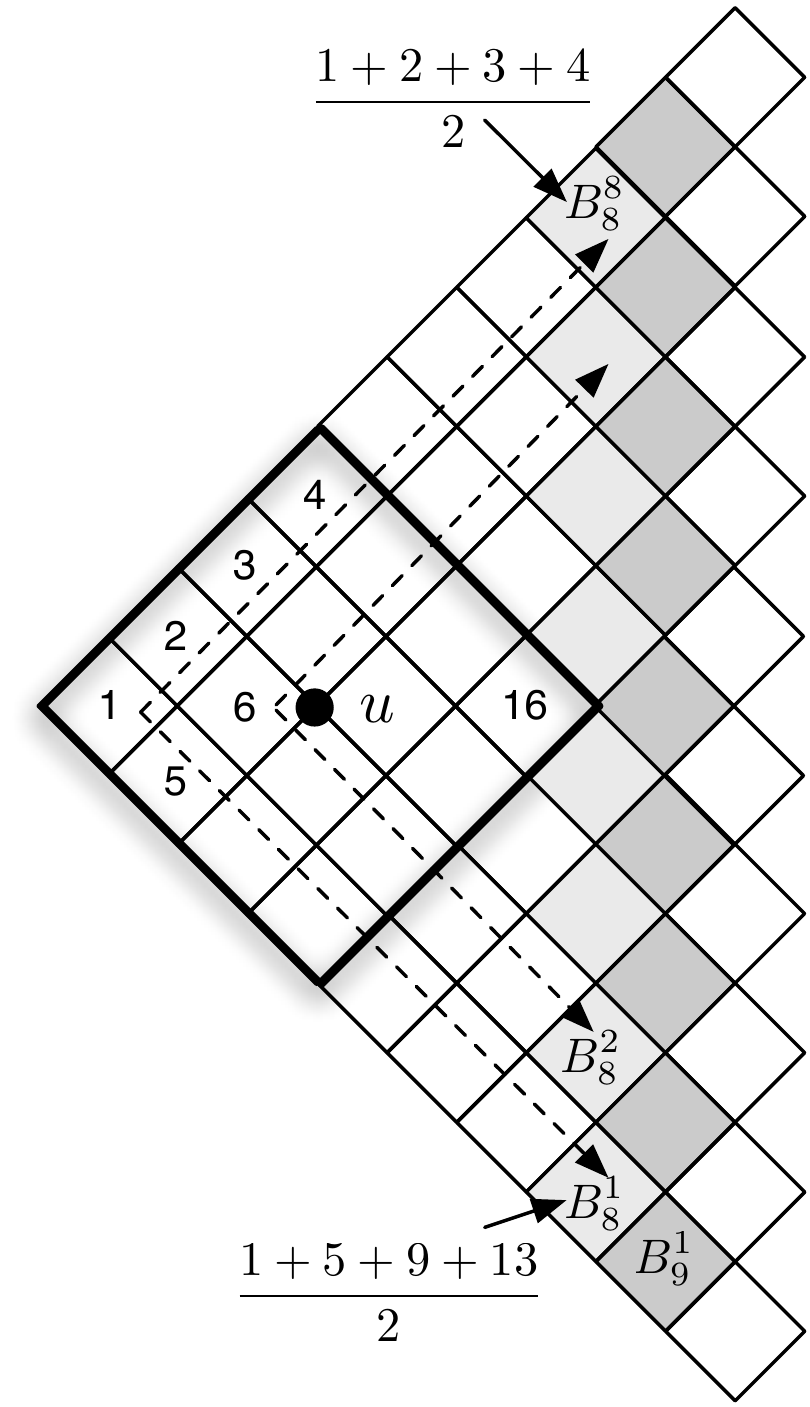}{The flow between the $d$ neighbors of $u$ to the first layer of the flow.}
\infig{.5}{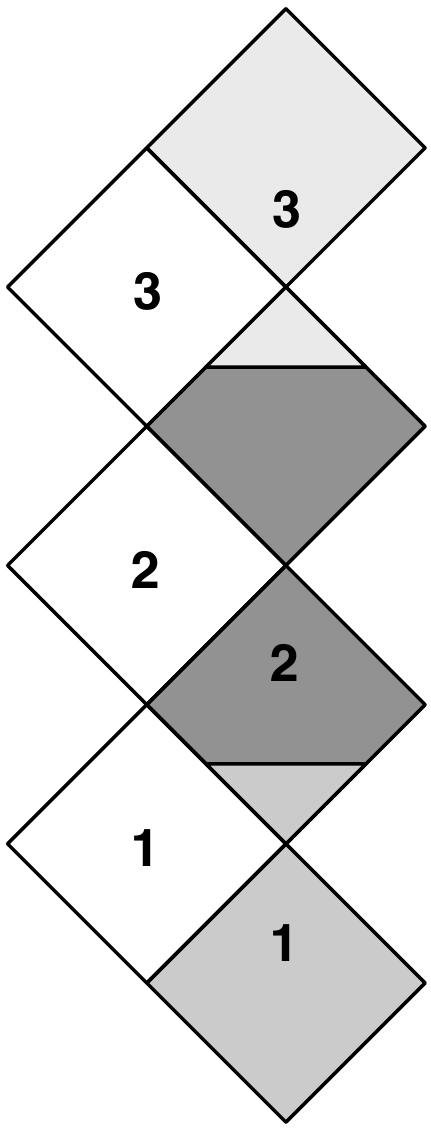}{An example of flow between two consecutive layers.}

\subsection{Proof of theorem \ref{thm:2D}}
\begin{proof}[of Theorem \ref{thm:2D}]
The results follows directly from Lemma \ref{lem:hitting2D} and Theorem \ref{thm:matthew}.
\end{proof}

\section{Conclusions}\label{sec:conc}
In this paper, we study the theoretical properties of simple random walks on wireless radio networks and simple random walks on hyper-graphs. The
techniques we developed in this paper shows that the cover time of a hyper-graph can be exponentially bigger than the cover time of a graph. This
suggests that one way to avoid this exponential time problem in wireless  networks is to have only one hyper-edge per node. We show that a general bound on the cover
time also holds for hyper-graphs, i.e., the cover time is less than $m \cdot n \cdot r(\cH)$. We also show that the radio hitting time can be computed in
polynomial time in hyper-graphs. We believe that random walks on hyper-graphs/radio networks will play an important role in data mining and in wireless
networks in the near future, as simple random walks on graphs did in the past.

\bibliographystyle{acm}
\bibliography{hyper,bsf}

\newpage
\appendix
\section*{APPENDIX: Proofs}
\section{Known Results on Random Walks on Graphs}
Let $G=(V,E)$ be an undirected graph with $|V|=n$ vertices with
$|E|=m$ edges. Let $N(G)$ be the electrical network having a node for each vertex in $V$, and, for every edge in $E$,
having a one ohm resistor between the corresponding nodes in $N(G)$.
Throughout this paper, we will abuse this notation and denote the electric network as the graph $G$ instead of the electrical network $N(G)$.
Let $R_{vu}$ be the effective resistance between the two node $v,u \in V$.
In \cite{chandra89electrical} theorem 2.1, the commute time between every pair of nodes $v,u \in V$ is:
\begin{align}\label{eq:Ch1}
C_{vu}=2mR_{vu}
\end{align}
Let $R=\max_{v,u\in V}R_{vu}$.
Let $N'(G)$ be an edge-weighted complete graph
having a vertex $v'$ for every vertex $v$ in $V$, and having an edge $(v',u')$ of weighted $R_{vu}$ for each pair of
(not necessarily adjacent) vertices $v,u$ in $V$.
Let $R^*$ be the weight of the minimum spanning tree in $N'(G)$.
In the journal \cite{chandra89electrical} theorem 2.4, the maximum cover time $C_G$ is:
\begin{align}\label{eq:Ch2}
mR \leq C_G \leq 2m\cdot \min(R(1+\ln(n)),R^*)
\end{align}
For an electrical network with $r_{vu}$ resistance between the nodes,
in the journal \cite{tetali91random} theorem 6, we have the following result:
\begin{align}\label{eq:Ch3}
\sum\limits_{\{v,u\} \in E}\frac{R_{vu}}{r_{vu}}=n-1
\end{align}
We will extend our notation to multi-graphs.
Let $G=(V,E)$ be a multi-graph without self loops with $|V|=n$ vertices with $|E|=m$ edges.
We denote $N(G)$ as an electrical network having a node for each vertex in $V$, and, for every edge in $E$,
having one ohm resistor between the corresponding nodes in $N(G)$.
We will abuse the notation and denote $G$ as the electrical network in the same way as in normal graph.

Equation~\ref{eq:Ch1} can be used for an electric network on a multi-graph,
and this can be proven in the same way as the proof for the simple graph.

\section{Proof of Lemma \ref{lem:lowermnr}}

\noindent \textbf{Lemma \ref{lem:lowermnr}.} \emph{
For $c\in \mathbb{N}, c \le n/2$ there exist a hyper-graph $\cH=(V,\mathcal{E}),|V|=n,|\mathcal{E}|=m$, $r(\cH)=c$ with
$C \ge \RC \ge \rh = \Omega(mn\cdot c)$.}

\begin{proof}
We build a hyper-graph in the following way:\\ We build a uniform hyper-clique(n',c) $\mathcal{C}=(V',\mathcal{E}')$.
We add a uniform hyper-line(n',2) $\mathcal{L}=(V'',\mathcal{E}'')$.
We join the hyper-line and the hyper-clique at a node $s$.
The node $s$ is a node at one of the ends of the hyper-line;
it will be a node in the hyper-line $\mathcal{L}$ and in the hyper-clique $\mathcal{C}$.
The hyper-graph will be the following $\cH=(V,\mathcal{E}),V=V'\cup V'',\mathcal{E}=\mathcal{E}' \cup \mathcal{E}''$.
The hyper-graph $H$ will have $n=|V|=2n'-1$ nodes and
$m=|\mathcal{E}|$ edges.
We use the bi-partite graph $B(\cH)=(V\cup \cE,E_B)$.
We use the same way as lemma ~\ref{l:l20} on the
commute time between the two ends of the uniform hyper-line $\cL$ to get a lower bound of $\Omega(mn\cdot c)$
on the radio cover time and radio hitting time of the hyper-graph.
\end{proof}

%

\end{document}